\documentclass[a4paper]{article}
\usepackage{geometry}
\usepackage{amsmath}
\usepackage{amssymb}
\usepackage{indentfirst}
\usepackage{mathdots}
\usepackage{cite}
\usepackage[colorlinks, linkcolor=blue, citecolor=blue]{hyperref}
\usepackage{mathrsfs}
\usepackage{booktabs}
\usepackage{enumitem}
\usepackage{authblk}
\usepackage{amsfonts,ntheorem}
\usepackage{threeparttable}
\usepackage{bm}
\usepackage{bbm, dsfont}
\usepackage{subfigure, epsfig}
\usepackage{setspace}
\usepackage{makecell}
\numberwithin{equation}{section}

\usepackage{cases}
\usepackage{algorithm}
\usepackage{algpseudocode}
\usepackage{multirow}

\makeatletter

\newcommand{\Rmnum}[1]{\expandafter\@slowromancap\romannumeral #1@}
\makeatother

\newtheorem{theorem}{Theorem}[section]
\newtheorem{lemma}[theorem]{Lemma}
\newtheorem{corollary}[theorem]{Corollary}

{
\theoremstyle{plain}
\theorembodyfont{\normalfont}

\newtheorem{definition}[theorem]{Definition}

\newtheorem{remark}[theorem]{Remark}

}

\newenvironment{proof}{\noindent{\textbf{\emph{Proof.}}}}

\allowdisplaybreaks[4]

\begin{document}

\title{\Large Special matrices over finite fields and their applications to quantum error-correcting codes}

\author[1]{{\small{Meng Cao}} \thanks{E-mail address: mengcaomath@126.com}}
\affil[1]{\footnotesize{Yanqi Lake Beijing Institute of Mathematical Sciences and Applications (BIMSA), Beijing, 101408, China} }
\renewcommand*{\Affilfont}{\small\it}

\date{}
\maketitle

{\linespread{1.4}{

\vskip -6mm

\noindent {\small{{\bfseries{Abstract:}}
The matrix-product (MP) code $\mathcal{C}_{A,k}:=[\mathcal{C}_{1},\mathcal{C}_{2},\ldots,\mathcal{C}_{k}]\cdot A$
with a non-singular by column (NSC) matrix $A$ plays an important role in constructing good quantum error-correcting codes.
In this paper, we study the MP code when the defining matrix $A$ satisfies the condition that $AA^{\dag}$ is $(D,\tau)$-monomial.
We give an explicit formula for calculating the dimension of the Hermitian hull of a MP code.
We provide the necessary and sufficient conditions that a MP code is Hermitian dual-containing (HDC), almost Hermitian dual-containing (AHDC),
Hermitian self-orthogonal (HSO), almost Hermitian self-orthogonal (AHSO), and Hermitian LCD, respectively.
We theoretically determine the number of all possible ways involving the relationships among the constituent codes to yield a MP code
with these properties, respectively.
We give alternative necessary and sufficient conditions for a MP code to be AHDC and AHSO, respectively,
and show several cases where a MP code is not AHDC or AHSO.
We provide the construction methods of HDC and AHDC MP codes, including those with optimal minimum distance lower bounds.
}}

\vspace{6pt}


\vspace{6pt}
}

\section{Introduction}\label{section1}
Denote by $\mathbb{F}_{q}$ (resp. $\mathbb{F}_{q^{2}}$) denote a finite field with $q$ (resp. $q^{2}$) elements, where $q$ is a prime power.
Denote by $\mathcal{M}_{k\times t}(\mathbb{F}_{q})$ (resp. $\mathcal{M}_{k\times t}(\mathbb{F}_{q^{2}})$) the set of all $k\times t$ matrices over
$\mathbb{F}_{q}$ (resp. $\mathbb{F}_{q^{2}}$).
It is denoted as $\mathcal{M}_{k}(\mathbb{F}_{q})$ (resp. $\mathcal{M}_{k}(\mathbb{F}_{q^{2}})$) if $k=t$.
Matrix-product (MP) codes are a type of long classical codes formed by combining several commensurate constituent codes
with a defining matrix.

\begin{definition}\label{definition1.1}
(\!\!\cite{O2002Note,Blackmore2001Matrix-product})
Let $\mathcal{C}_{1},\mathcal{C}_{2},\ldots,\mathcal{C}_{k}$ be $q$-ary linear codes of length $n$.
Let $A=(a_{i,j})\in \mathcal{M}_{k\times t}(\mathbb{F}_{q})$ with $k\leq t$.
A \emph{matrix-product (MP) code}, denoted as
$\mathcal{C}(A):=[\mathcal{C}_{1},\mathcal{C}_{2},\ldots,\mathcal{C}_{k}]\cdot A$,
is defined as the set of all matrix-products $[\mathbf{c}_{1},\mathbf{c}_{2},\ldots,\mathbf{c}_{k}]\cdot A$ with $\mathbf{c}_{i}=(c_{1,i},c_{2,i},\ldots,c_{n,i})^{\top}\in\mathcal{C}_{i}$ for $i=1,2,\ldots,k$, where
$A$ is called the \emph{defining matrix}, and $\mathcal{C}_{1},\mathcal{C}_{2},\ldots,\mathcal{C}_{k}$ are called the \emph{constituent codes}.
\end{definition}

A typical codeword $\mathbf{c}=[\mathbf{c}_{1},\mathbf{c}_{2},\ldots,\mathbf{c}_{k}]\cdot A$ of $\mathcal{C}(A)$ is an $n\times t$ matrix with the following form:
\begin{align*}
\mathbf{c}=\left(
\begin{array}{cccc}
\sum_{i=1}^k c_{1,i}a_{i,1}& \sum_{i=1}^k c_{1,i}a_{i,2}& \cdots &\sum_{i=1}^k c_{1,i}a_{i,t}\\
\sum_{i=1}^k c_{2,i}a_{i,1}& \sum_{i=1}^k c_{2,i}a_{i,2}& \cdots &\sum_{i=1}^k c_{2,i}a_{i,t}\\
\vdots&\vdots&\ddots&\vdots\\
\sum_{i=1}^k c_{n,i}a_{i,1}& \sum_{i=1}^k c_{n,i}a_{i,2}& \cdots &\sum_{i=1}^k c_{n,i}a_{i,t}\\
\end{array}\right).
\end{align*}
Reading the entries of the $n\times t$ matrix above in column-major order, any codeword of $\mathcal{C}(A)$ can be viewed as a row vector of length $tn$, i.e.,
$\mathbf{c}=\big[\sum_{i=1}^k a_{i,1}\mathbf{c}_{i},\sum_{i=1}^k a_{i,2}\mathbf{c}_{i},\ldots,\sum_{i=1}^k a_{i,t}\mathbf{c}_{i}\big]$,
where $\mathbf{c}_{i}=(c_{1,i},c_{2,i},\ldots,c_{n,i})\in\mathcal{C}_{i}$ is an $1\times n$ row vector for $i=1,2,\ldots,k$.
For simplicity, we use the notation $\mathcal{C}_{A,k}$ to represent the MP code $\mathcal{C}(A)$ when $A\in \mathcal{M}_{k}(\mathbb{F}_{q})$.

Scholars have increasingly focused on MP codes in recent years, with numerous intriguing properties of these codes being unveiled,
such as their minimum distances, duals, hulls, squares, decoding algorithms, etc.
It was shown in \cite{Blackmore2001Matrix-product,Hernando2009Construction,O2002Note} that the minimum distance (or minimum distance lower bound) of a MP code
is related to the properties of its defining matrix and constituent codes.
The Euclidean dual and Hermitian dual of a MP code were shown in \cite{Blackmore2001Matrix-product} and \cite{Zhang2015qc}, respectively.
In \cite{Liu2020Galois}, Liu and Pan investigated the $\ell$-Galois hull of MP codes over finite fields.
In \cite{Cao2023On}, Cao \emph{et al}. studied the $\sigma$ duals, $\sigma$ duals and intersections of MP codes over finite fields.
Cascudo \emph{et al}. \cite{Cascudo2020Squares} studied the squares of MP codes over finite fields.
In \cite{Hernando2009Construction,Hernando2012List,Hernando2013Decoding}, Hernando \emph{et al}. proposed the decoding algorithms for MP codes.
MP codes over different rings were explored in \cite{Fan2014Matrix,Fan2014Homogeneous,Asch2008Matrix,Cao2018Matrix}.
MP codes are very useful for constructing quantum codes (see, e.g., \cite{Galindo2015New,Cao2024On,Liu2018On})
and locally repairable codes (see, e.g., \cite{Chen2023New,Luo2022Three}).

In this paper, we aim to investigate some properties of MP code $\mathcal{C}_{A,k}$ whose defining matrix $A$ satisfies the condition that $AA^{\dag}$ is a $(D,\tau)$-monomial matrix, namely, $AA^{\dag}=DP_{\tau}$ for an invertible diagonal matrix $D$ and a permutation matrix $P_{\tau}$ under the permutation $\tau\in S _{k}$, where $S _{k}$ denotes the symmetric group of order $k$.

The remainder of this paper is organized as follows.
In Section 2, we give the preliminaries of this paper.
In Sections 3 and 4, we present an explicit formula for calculating the dimension of the Hermitian hull of $\mathcal{C}_{A,k}$.
Based on this, we present the necessary and sufficient conditions for a MP code $\mathcal{C}_{A,k}$ to be HDC, AHDC, HSO, AHSO and Hermitian LCD respectively.
From this framework, we theoretically determine the number of all possible ways involving the relationships among the constituent codes $\mathcal{C}_{1},\mathcal{C}_{2},\ldots,\mathcal{C}_{k}$ to yield a MP code $\mathcal{C}_{A,k}$ that is HDC, AHDC, HSO and AHSO, respectively.
We also give alternative necessary and sufficient conditions for $\mathcal{C}_{A,k}$ to be AHDC and AHSO, respectively,
and provide several cases where $\mathcal{C}_{A,k}$ is not AHDC or AHSO.
In Section 5, we give the construction methods of HDC and AHDC MP codes, including those with optimal minimum distance lower bounds.
In Section 6, we give the conclusion of this paper.

\section{Preliminaries}\label{section2}

For any matrix $A=[a_{i,j}]$ over $\mathbb{F}_{q^{2}}$, denote by $A^{(q)}:=\big[a_{i,j}^{q}\big]$ the \emph{conjugate matrix} of $A$
and denote by $A^{\dag}:=\big(A^{(q)}\big)^{\top}=\big[a_{j,i}^{q}\big]$ the \emph{transposed conjugate matrix} of $A$.

Let $[n,k,d]_{q^{2}}$ denote a $q^{2}$-ary linear code of length $n$, dimension $k$ and minimum distance $d$.
It is a $k$-dimensional subspace of the vector space $\mathbb{F}_{q^{2}}^{n}$, which satisfies the \emph{Singleton bound}, i.e., $d\leq n-k+1$.
If $d=n-k+1$, then it is called a \emph{maximum distance separable (MDS) code}.

For vectors $\mathbf{x}=(x_{1},x_{2},\ldots,x_{n})\in\mathbb{F}_{q^{2}}^{n}$ and $\mathbf{y}=(y_{1},y_{2},\ldots,y_{n})\in\mathbb{F}_{q^{2}}^{n}$,
their \emph{(Euclidean) inner product} and \emph{Hermitian inner product} is defined by
$\langle\mathbf{x},\mathbf{y}\rangle:=\sum_{i=1}^n x_i y_i$ and
$\langle\mathbf{x},\mathbf{y}\rangle_{H}:=\sum_{i=1}^n x_i y_i^{q}$, respectively.
The \emph{(Euclidean) dual code} and \emph{Hermitian dual code} of a liner code $\mathcal{C}$ of length $n$ over $\mathbb{F}_{q^{2}}$ is defined by
$\mathcal{C}^{\perp}:=\{\mathbf{x}\in\mathbb{F}_{q^{2}}^{n}|\langle\mathbf{x},\mathbf{y}\rangle=0 \ \mathrm{for} \ \mathrm{all} \  \mathbf{y}\in \mathcal{C}\}$ and
$\mathcal{C}^{\perp_{H}}:=\{\mathbf{x}\in\mathbb{F}_{q^{2}}^{n}|\langle\mathbf{x},\mathbf{y}\rangle_{H}=0 \ \mathrm{for} \ \mathrm{all} \  \mathbf{y}\in \mathcal{C}\}$, respectively.

We fix the following notations in the rest of this paper.
\vspace{-4pt}
\begin{itemize}
\setlength{\itemsep}{1pt}
\setlength{\parsep}{1pt}
\setlength{\parskip}{1pt}
\item $\mathbb{F}_{q}^{\ast}:=\mathbb{F}_{q}\backslash \{0\}$, $\mathbb{F}_{q^{2}}^{\ast}:=\mathbb{F}_{q^{2}}\backslash \{0\}$.

\item $A^{\top}$ denotes the transposed matrix of the matrix $A$.
For $A=(a_{i,j})$ over $\mathbb{F}_{q^{2}}$, denote by $A^{\dag}:=\big[a_{j,i}^{q}\big]$ the conjugate transposed matrix of $A$.

\item  A $k\times k$ matrix $A$ is called $\tau$-monomial if it can be written as $A=DP_{\tau}$, where $D$ is a invertible diagonal matrix, and
$P_{\tau}$ is a permutation matrix under $\tau$ satisfying that the $\tau(i)$-th row of $P_{\tau}$ is exactly the $i$-th row of the identity matrix $I_{k}$
for $i=1,\ldots,k$.

\item $J_{t}$ denotes the $t\times t$ flip matrix whose $(i,j)$-th entry is $1$ for all $i+j=t+1$ with $1\leq i,j\leq t$, and entry $0$, otherwise.

\item $\mathrm{Hull}_{H}(\mathcal{C}):=\mathcal{C}\cap\mathcal{C}^{\bot_{H}}$ is called the \emph{Hermitian hull} of the linear code $\mathcal{C}$.
\end{itemize}

\section{Several types of MP codes}
\begin{lemma}
(\!\!\cite{Guenda2020Linear})
Let $\mathcal{C}_{i}$ be an $[n,t_{i},d_{i}]_{q^{2}}$ code with a generator matrix $G_{i}$ for $i=1,2$. Then
$\mathrm{dim}(\mathcal{C}_{1}\cap\mathcal{C}_{2}^{\bot_{H}})=t_{1}-\mathrm{rank}(G_{1}G_{2}^{\dag})$.
\end{lemma}

\begin{corollary}
(\!\!\cite{Guenda2018})
Let $\mathcal{C}$ be an $[n,t,d]_{q^{2}}$ code with a generator matrix $G$. Then
$\mathrm{dim}(\mathrm{Hull}_{H}(\mathcal{C}))=t-\mathrm{rank}(GG^{\dag})$.
\end{corollary}

\begin{theorem} Let $\mathcal{C}_{i}$ be an $[n,t_{i},d_{i}]_{q^{2}}$ code with a generator matrix $G_{i}$ for $i=1,2,\ldots,k$.
Let $\mathcal{C}_{A,k}:=[\mathcal{C}_{1},\mathcal{C}_{2},\ldots,\mathcal{C}_{k}]\cdot A$ be a MP code with a generator matrix $G$,
where $A\in\mathcal{M}(\mathbb{F}_{q^{2}},k)$.
If $AA^{\dag}=DP_{\tau}$ is monomial, then
$\mathrm{dim}(\mathrm{Hull}_{H}(\mathcal{C}_{A,k}))=\sum_{i=1}^{k}\mathrm{dim}\big(\mathcal{C}_{i}\cap\mathcal{C}_{\tau(i)}^{\bot_{H}}\big)$.
\end{theorem}

\begin{proof}
Let $\mathbf{a}_{i}$ denote the $i$-th row of $A$ for $i=1,2,\ldots,k$. Then
$$G=\left(
\begin{array}{c}
\mathbf{a}_{1}\otimes G_{1}\\
\mathbf{a}_{2}\otimes G_{2}\\
\vdots\\
\mathbf{a}_{k}\otimes G_{k}\\
\end{array}\right),$$
which, together with $AA^{\dag}=DP_{\tau}$, derives that
\begin{align*}
GG^{\dag}
&=\left(
\begin{array}{c}
\mathbf{a}_{1}\otimes G_{1}\\
\mathbf{a}_{2}\otimes G_{2}\\
\vdots\\
\mathbf{a}_{k}\otimes G_{k}\\
\end{array}\right)(\mathbf{a}_{1}^{\dag}\otimes G_{1}^{\dag},\mathbf{a}_{2}^{\dag}\otimes G_{2}^{\dag},\ldots,\mathbf{a}_{k}^{\dag}\otimes G_{k}^{\dag})\\
&=\left(
\begin{array}{cccc}
\langle\mathbf{a}_{1},\mathbf{a}_{1}\rangle_{H}G_{1}G_{1}^{\dag}&\langle\mathbf{a}_{1},\mathbf{a}_{2}\rangle_{H}G_{1}G_{2}^{\dag}
&\cdots&\langle\mathbf{a}_{1},\mathbf{a}_{k}\rangle_{H}G_{1}G_{k}^{\dag}\\
\langle\mathbf{a}_{2},\mathbf{a}_{1}\rangle_{H}G_{2}G_{1}^{\dag}&\langle\mathbf{a}_{2},\mathbf{a}_{2}\rangle_{H}G_{2}G_{2}^{\dag}
&\cdots&\langle\mathbf{a}_{2},\mathbf{a}_{k}\rangle_{H}G_{2}G_{k}^{\dag}\\
\vdots&\vdots&\ddots&\vdots\\
\langle\mathbf{a}_{k},\mathbf{a}_{1}\rangle_{H}G_{k}G_{1}^{\dag}&\langle\mathbf{a}_{k},\mathbf{a}_{2}\rangle_{H}G_{k}G_{2}^{\dag}
&\cdots&\langle\mathbf{a}_{k},\mathbf{a}_{k}\rangle_{H}G_{k}G_{k}^{\dag}\\
\end{array}\right)\\
&=\left(
\begin{array}{cccc}
\cdots&\langle\mathbf{a}_{1},\mathbf{a}_{\tau(1)}\rangle_{H}G_{1}G_{\tau(1)}^{\dag}&\cdots&\cdots\\
\vdots&\vdots&\ddots&\vdots\\
\cdots&\cdots&\langle\mathbf{a}_{k},\mathbf{a}_{\tau(k)}\rangle_{H}G_{k}G_{\tau(k)}^{\dag}&\cdots\\
\end{array}\right).
\end{align*}
Thus, by Lemma 3.1, we obtain
\begin{align*}
\mathrm{rank}(GG^{\dag})
=\sum_{i=1}^{k}\mathrm{rank}(G_{i}G_{\tau(i)}^{\dag})
=\sum_{i=1}^{k}\left(t_{i}-\mathrm{dim}\big(\mathcal{C}_{i}\cap\mathcal{C}_{\tau(i)}^{\bot_{H}}\big)\right).
\end{align*}
Combing this with Corollary 3.2, we obtain
$$\mathrm{dim}(\mathrm{Hull}_{H}(\mathcal{C}_{A,k}))=\sum_{i=1}^{k}t_{i}-\mathrm{rank}(GG^{\dag})
=\sum_{i=1}^{k}\mathrm{dim}\big(\mathcal{C}_{i}\cap\mathcal{C}_{\tau(i)}^{\bot_{H}}\big),$$
which completes the proof. $\hfill\square$
\end{proof}

\vspace{6pt}

In the following theorem, we are ready to give the necessary and sufficient conditions that a MP code is HDC, AHDC, HSO, AHSO and Hermitian LCD, respectively.

\begin{theorem}
Let $AA^{\dag}=DP_{\tau}$ be monomial for $A\in\mathcal{M}(\mathbb{F}_{q^{2}},k)$.
Let $\mathcal{C}_{A,k}:=[\mathcal{C}_{1},\mathcal{C}_{2},\ldots,\mathcal{C}_{k}]\cdot A$, where
$\mathcal{C}_{\ell}$ is an $[n,t_{\ell},d_{\ell}]_{q^{2}}$ code for $\ell\in\{1,\ldots,k\}$.
Then the following statements hold.
\vspace{-4pt}
\begin{itemize}
\item [1)] $\mathcal{C}_{A,k}$ is HDC if and only if $\mathcal{C}_{\tau(i)}^{\bot_{H}}\subseteq\mathcal{C}_{i}$ for all $i\in\{1,\ldots,k\}$.

\vspace{-4pt}

\item [2)] $\mathcal{C}_{A,k}$ is AHDC if and only if
\vspace{-8pt}
\begin{align}
\sum_{i=1}^{k}\big(\mathrm{dim}(\mathcal{C}_{\tau(i)}^{\bot_{H}})-\mathrm{dim}(\mathcal{C}_{i}\cap\mathcal{C}_{\tau(i)}^{\bot_{H}})\big)=1
\end{align}
if and only if there exists a unique $j\in\{1,\ldots,k\}$ such that
$\mathrm{dim}(\mathcal{C}_{\tau(j)}^{\bot_{H}})-\mathrm{dim}(\mathcal{C}_{j}\cap\mathcal{C}_{\tau(j)}^{\bot_{H}})=1$
and $\mathcal{C}_{\tau(i)}^{\bot_{H}}\subseteq\mathcal{C}_{i}$ for all $i\in\{1,\ldots,k\}\backslash\{j\}$.

\vspace{-4pt}

\item [3)] $\mathcal{C}_{A,k}$ is HSO if and only if $\mathcal{C}_{i}\subseteq\mathcal{C}_{\tau(i)}^{\bot_{H}}$ for all $i\in\{1,\ldots,k\}$.

\vspace{-4pt}

\item [4)] $\mathcal{C}_{A,k}$ is AHSO if and only if
\vspace{-8pt}
\begin{align}
\sum_{i=1}^{k}\big(\mathrm{dim}(\mathcal{C}_{i})-\mathrm{dim}(\mathcal{C}_{i}\cap\mathcal{C}_{\tau(i)}^{\bot_{H}})\big)=1
\end{align}
if and only if there exists a unique $j\in\{1,\ldots,k\}$ such that
$\mathrm{dim}(\mathcal{C}_{j})-\mathrm{dim}(\mathcal{C}_{j}\cap\mathcal{C}_{\tau(j)}^{\bot_{H}})=1$
and $\mathcal{C}_{i}\subseteq\mathcal{C}_{\tau(i)}^{\bot_{H}}$ for all $i\in\{1,\ldots,k\}\backslash\{j\}$.

\vspace{-4pt}

\item [5)] $\mathcal{C}_{A,k}$ is Hermitian LCD if and only if
$\mathcal{C}_{i}\cap\mathcal{C}_{\tau(i)}^{\bot_{H}}=\{\mathbf{0}\}$ for all $i\in\{1,\ldots,k\}$.
\end{itemize}
\end{theorem}

\begin{proof}
1) We have
\vspace{-6pt}
\begin{align*}
\mathcal{C}_{A,k}^{\bot_{H}}\subseteq\mathcal{C}_{A,k}
&\Longleftrightarrow\sum_{i=1}^{k}\mathrm{dim}\big(\mathcal{C}_{i}\cap\mathcal{C}_{\tau(i)}^{\bot_{H}}\big)
=\mathrm{dim}(\mathcal{C}_{A,k}^{\bot_{H}})=\sum_{i=1}^{k}\mathrm{dim}(\mathcal{C}_{\tau(i)}^{\bot_{H}}).\\
&\Longleftrightarrow\mathcal{C}_{i}\cap\mathcal{C}_{\tau(i)}^{\bot_{H}}=\mathcal{C}_{\tau(i)}^{\bot_{H}}, \ 1\leq i\leq k.\\
&\Longleftrightarrow\mathcal{C}_{\tau(i)}^{\bot_{H}}\subseteq\mathcal{C}_{i}, \ 1\leq i\leq k.
\end{align*}
This proves statement 1).

2) We deduce that
\vspace{-8pt}
\begin{align*}
\mathcal{C}_{A,k}\ \text{is AHDC}
&\Longleftrightarrow\mathrm{dim}(\mathrm{Hull}_{H}(\mathcal{C}_{A,k}))
=\mathrm{dim}(\mathcal{C}_{A,k}^{\bot_{H}})-1.\\
&\Longleftrightarrow\sum_{i=1}^{k}\mathrm{dim}\big(\mathcal{C}_{i}\cap\mathcal{C}_{\tau(i)}^{\bot_{H}}\big)
=\mathrm{dim}(\mathcal{C}_{A,k}^{\bot_{H}})-1.\\
&\Longleftrightarrow\sum_{i=1}^{k}\big(\mathrm{dim}(\mathcal{C}_{\tau(i)}^{\bot_{H}})-\mathrm{dim}(\mathcal{C}_{i}\cap\mathcal{C}_{\tau(i)}^{\bot_{H}})\big)=1.\\
&\Longleftrightarrow\mathrm{dim}(\mathcal{C}_{\tau(j)}^{\bot_{H}})-\mathrm{dim}(\mathcal{C}_{j}\cap\mathcal{C}_{\tau(j)}^{\bot_{H}})=1 \ \text{for a unique}\ j\in\{1,\ldots,k\},
\end{align*}
\vspace{-40pt}
\begin{align*}
\ \ \ \ \ \ \ \ \ \ \ \ \ \ \ \ \ \ \ \ \ \ \ \ \ \ \ \ \ \ \ \ \ \ \ \ \ \ \ \text{and} \ \mathcal{C}_{\tau(i)}^{\bot_{H}}\subseteq\mathcal{C}_{i}\
\text{for all}\
i\in\{1,\ldots,k\}\backslash\{j\}.
\end{align*}
This completes statement 2).

3) We have
\vspace{-6pt}
\begin{align*}
\mathcal{C}_{A,k}\subseteq\mathcal{C}_{A,k}^{\bot_{H}}
&\Longleftrightarrow\sum_{i=1}^{k}\mathrm{dim}\big(\mathcal{C}_{i}\cap\mathcal{C}_{\tau(i)}^{\bot_{H}}\big)
=\mathrm{dim}(\mathcal{C}_{A,k})=\sum_{i=1}^{k}\mathrm{dim}(\mathcal{C}_{i}).\\
&\Longleftrightarrow\mathcal{C}_{i}\cap\mathcal{C}_{\tau(i)}^{\bot_{H}}=\mathcal{C}_{i}, \ 1\leq i\leq k.\\
&\Longleftrightarrow\mathcal{C}_{i}\subseteq\mathcal{C}_{\tau(i)}^{\bot_{H}}, \ 1\leq i\leq k.
\end{align*}
This proves statement 3).

4) We deduce that
\vspace{-8pt}
\begin{align*}
\mathcal{C}_{A,k}\ \text{is AHSO}
&\Longleftrightarrow\mathrm{dim}(\mathrm{Hull}_{H}(\mathcal{C}_{A,k}))
=\mathrm{dim}(\mathcal{C}_{A,k})-1.\\
&\Longleftrightarrow\sum_{i=1}^{k}\mathrm{dim}\big(\mathcal{C}_{i}\cap\mathcal{C}_{\tau(i)}^{\bot_{H}}\big)
=\mathrm{dim}(\mathcal{C}_{A,k})-1.\\
&\Longleftrightarrow\sum_{i=1}^{k}\big(\mathrm{dim}(\mathcal{C}_{i})-\mathrm{dim}(\mathcal{C}_{i}\cap\mathcal{C}_{\tau(i)}^{\bot_{H}})\big)=1.\\
&\Longleftrightarrow\mathrm{dim}(\mathcal{C}_{j})-\mathrm{dim}(\mathcal{C}_{j}\cap\mathcal{C}_{\tau(j)}^{\bot_{H}})=1 \ \text{for a unique}\ j\in\{1,\ldots,k\},
\end{align*}
\vspace{-40pt}
\begin{align*}
\ \ \ \ \ \ \ \ \ \ \ \ \ \ \ \ \ \ \ \ \ \ \ \ \ \ \ \ \ \ \ \ \ \ \ \ \ \ \ \text{and} \ \mathcal{C}_{i}\subseteq\mathcal{C}_{\tau(i)}^{\bot_{H}}\ \text{for all}\ i\in\{1,\ldots,k\}\backslash\{j\}.
\end{align*}
This completes statement 4).

5) Statement 5) follows directly. $\hfill\square$
\end{proof}

\vspace{6pt}

For the trivial case, i.e., $\tau=\mathbbm{1}_{k}$, it is easy to obtain the following corollary.

\begin{corollary}
With the notations in Theorem 3.4.
Let $AA^{\dag}=D$ be invertible diagonal for $A\in\mathcal{M}(\mathbb{F}_{q^{2}},k)$. Then the following statements hold.
\vspace{-4pt}
\begin{itemize}
\item [1)] $\mathcal{C}_{A,k}$ is HDC if and only if $\mathcal{C}_{i}$ is HDC for all $i\in\{1,\ldots,k\}$.

\vspace{-4pt}

\item [2)] $\mathcal{C}_{A,k}$ is AHDC if and only if
\vspace{-8pt}
\begin{align*}
\sum_{i=1}^{k}\big(\mathrm{dim}(\mathcal{C}_{i}^{\bot_{H}})-\mathrm{dim}(\mathrm{Hull}_{H}(\mathcal{C}_{i}))\big)=1
\end{align*}
if and only if there exists a unique $j\in\{1,\ldots,k\}$ such that $\mathcal{C}_{j}$ is AHDC and $\mathcal{C}_{i}$ is HDC for all $i\in\{1,\ldots,k\}\backslash\{j\}$.

\vspace{-4pt}

\item [3)] $\mathcal{C}_{A,k}$ is HSO if and only if $\mathcal{C}_{i}$ is HSO for all $i\in\{1,\ldots,k\}$.

\vspace{-4pt}

\item [4)] $\mathcal{C}_{A,k}$ is AHSO if and only if
\vspace{-8pt}
\begin{align*}
\sum_{i=1}^{k}\big(\mathrm{dim}(\mathcal{C}_{i})-\mathrm{dim}(\mathrm{Hull}_{H}(\mathcal{C}_{i}))\big)=1
\end{align*}
if and only if there exists a unique $j\in\{1,\ldots,k\}$ such that $\mathcal{C}_{j}$ is AHSO and $\mathcal{C}_{i}$ is HSO for all $i\in\{1,\ldots,k\}\backslash\{j\}$.

\vspace{-4pt}

\item [5)] $\mathcal{C}_{A,k}$ is Hermitian LCD if and only if
$\mathcal{C}_{i}$ is Hermitian LCD for all $i\in\{1,\ldots,k\}$.
\end{itemize}
\end{corollary}

\section{Further characterization of Theorem 3.4 for $\tau\neq\mathbbm{1}_{k}$}

\subsection{An explicit characterization of Theorem 3.4 for $\tau\neq\mathbbm{1}_{k}$}

\begin{lemma}
Let $B\in\mathcal{M}(\mathbb{F}_{q^{2}},k)$. We have the following statements.
\vspace{-4pt}
\begin{itemize}
\item [(1)] If $B=B^{\top}=DP_{\tau}$ is monomial for $D=\mathrm{diag}(\mu_{1},\mu_{2},\ldots,\mu_{k})$ and $\tau\in S_{k}$,
then $\tau^{2}=\mathbbm{1}_{k}$ and $\mu_{\tau(i)}=\mu_{i}$ for all $i\in\{1,\ldots,k\}$.

\vspace{-4pt}

\item [(2)] If $B=B^{\dag}=DP_{\tau}$ is monomial for $D=\mathrm{diag}(\mu_{1},\mu_{2},\ldots,\mu_{k})$ and $\tau\in S_{k}$,
then $\tau^{2}=\mathbbm{1}_{k}$ and $\mu_{\tau(i)}=\mu_{i}^{q}$ for all $i\in\{1,\ldots,k\}$.
\end{itemize}
\end{lemma}

\begin{proof}
Let us prove statement (2).
Write $P_{\tau}=\sum_{i=1}^{k}E_{i,\tau_(i)}$ and
$D=\sum_{i=1}^{k}\mu_{i}E_{i,i}$. Then $B=\sum_{i=1}^{k}\mu_{i}E_{i,\tau(i)}$, which derives that
$B^{\dag}=\sum_{i=1}^{k}\mu_{i}^{q}E_{\tau(i),i}=\sum_{i=1}^{k}\mu_{\tau(i)}^{q}E_{\tau^{2}(i),\tau(i)}$.
As $B=B^{\dag}$, we obtain
\vspace{-8pt}
\begin{align*}
\sum_{i=1}^{k}\mu_{i}E_{i,\tau(i)}=\sum_{i=1}^{k}\mu_{\tau(i)}^{q}E_{\tau^{2}(i),\tau(i)}.
\end{align*}
Thus, $\mu_{\tau(i)}^{q}=\mu_{i}$, i.e., $\mu_{\tau(i)}=\mu_{i}^{q}$ for all $i\in\{1,\ldots,k\}$;
$\tau^{2}(i)=i$ for all $i\in\{1,\ldots,k\}$, i.e., $\tau^{2}=\mathbbm{1}_{k}$. Similarly, one can obtain statement (1).
$\hfill\square$
\end{proof}

\begin{corollary}
Let $A\in\mathcal{M}(\mathbb{F}_{q^{2}},k)$. We have the following statements.
\vspace{-4pt}
\begin{itemize}
\item [(1)] If $AA^{\top}=DP_{\tau}$ is monomial for $D=\mathrm{diag}(\mu_{1},\mu_{2},\ldots,\mu_{k})$ and $\tau\in S_{k}$,
then $\tau^{2}=\mathbbm{1}_{k}$ and $\mu_{\tau(i)}=\mu_{i}$ for all $i\in\{1,\ldots,k\}$.

\vspace{-4pt}

\item [(2)] If $AA^{\dag}=DP_{\tau}$ is monomial for $D=\mathrm{diag}(\mu_{1},\mu_{2},\ldots,\mu_{k})$ and $\tau\in S_{k}$,
then $\tau^{2}=\mathbbm{1}_{k}$ and $\mu_{\tau(i)}=\mu_{i}^{q}$ for all $i\in\{1,\ldots,k\}$.
\end{itemize}

\end{corollary}

\begin{remark}
For $\tau\in S_{k}$, it is easy to know $\tau^{2}=\mathbbm{1}_{k}$ if and only if one the following two cases hold:
\vspace{-3pt}
\begin{itemize}
\item (Trivial case) $\tau$ is the identity permutation, i.e., $\tau=\mathbbm{1}_{k}$.

\vspace{-3pt}
\item (Nontrivial case) $\tau$ is a permutation of order $2$, that is, $\tau$ can be expressed as the product of some disjoint 2-cycles.
\end{itemize}
\end{remark}

The following lemma calculates the number of the permutations $\tau$ of order 2 in $S_{k}$ and thus obtain the number of the permutations $\tau\in S_{k}$ satisfying $\tau^{2}=\mathbbm{1}_{k}$.

\begin{lemma}
(\!\!\cite{Cao2024On})
1) The number of the permutations $\tau$ of order 2 in $S_{k}$ is $\sum_{i=1}^{\lfloor\frac{k}{2}\rfloor}\frac{k!}{(k-2i)!\cdot i!\cdot2^{i}}$.

2) The number of the permutations $\tau\in S_{k}$ satisfying $\tau^{2}=\mathbbm{1}_{k}$ is $\sum_{i=0}^{\lfloor\frac{k}{2}\rfloor}\frac{k!}{(k-2i)!\cdot i!\cdot2^{i}}$.
\end{lemma}

\vspace{7pt}

We give a more explicit characterization of Theorem 3.4 for $\tau\neq\mathbbm{1}_{k}$ in the following Theorem.

\begin{theorem}
Let $AA^{\dag}=DP_{\tau}$ be monomial for $A\in\mathcal{M}(\mathbb{F}_{q^{2}},k)$,
where $\tau=(i_{1}\ j_{1})(i_{2}\ j_{2})\cdots(i_{s}\ j_{s})$ denotes the product of $s$ disjoint 2-cycles for $s\in\{1,\ldots,\lfloor\frac{k}{2}\rfloor\}$.
Let $\mathcal{F}=\{1,\ldots,k\}-\{i_{1},\ldots,i_{s},j_{1},\ldots,j_{s}\}$.
Let $\mathcal{C}_{A,k}:=[\mathcal{C}_{1},\mathcal{C}_{2},\ldots,\mathcal{C}_{k}]\cdot A$,
where $\mathcal{C}_{\ell}$ is an $[n,t_{\ell},d_{\ell}]_{q^{2}}$ code for $\ell\in\{1,\ldots,k\}$.
Then the following statements hold.
\vspace{-4pt}
\begin{itemize}
\item [1)] $\mathcal{C}_{A,k}$ is HDC if and only if $\mathcal{C}_{j_{r}}^{\bot_{H}}\subseteq\mathcal{C}_{i_{r}}$ for all $r\in\{1,\ldots,s\}$
and $\mathcal{C}_{i}$ is HDC for all $i\in\mathcal{F}$.

\vspace{-4pt}

\item [2)] $\mathcal{C}_{A,k}$ is AHDC if and only if
\vspace{-6pt}
\begin{align*}
\sum_{r=1}^{s}\big(\mathrm{dim}(\mathcal{C}_{j_{r}}^{\bot_{H}})-\mathrm{dim}(\mathcal{C}_{i_{r}}\cap\mathcal{C}_{j_{r}}^{\bot_{H}})\big)=0,\
\sum_{i\in\mathcal{F}}\big(\mathrm{dim}(\mathcal{C}_{i}^{\bot_{H}})-\mathrm{dim}(\mathrm{Hull}_{H}(\mathcal{C}_{i}))\big)=1
\end{align*}
if and only if $\mathcal{C}_{j_{r}}^{\bot_{H}}\subseteq\mathcal{C}_{i_{r}}$ for all $r\in\{1,\ldots,s\}$ and there exists a unique $j\in\mathcal{F}$ such that
$\mathcal{C}_{j}$ is AHDC and $\mathcal{C}_{i}$ is HDC for all $i\in\mathcal{F}\backslash\{j\}$.

\vspace{-4pt}

\item [3)] $\mathcal{C}_{A,k}$ is HSO if and only if $\mathcal{C}_{i_{r}}\subseteq\mathcal{C}_{j_{r}}^{\bot_{H}}$ for all $r\in\{1,\ldots,s\}$
and $\mathcal{C}_{i}$ is HSO for all $i\in\mathcal{F}$.

\vspace{-4pt}

\item [4)] $\mathcal{C}_{A,k}$ is AHSO if and only if
\vspace{-6pt}
\begin{align*}
\sum_{r=1}^{s}\big(\mathrm{dim}(\mathcal{C}_{i_{r}})-\mathrm{dim}(\mathcal{C}_{i_{r}}\cap\mathcal{C}_{j_{r}}^{\bot_{H}})\big)=0,\
\sum_{i\in\mathcal{F}}\big(\mathrm{dim}(\mathcal{C}_{i})-\mathrm{dim}(\mathcal{C}_{i}\cap\mathcal{C}_{i}^{\bot_{H}})\big)=1
\end{align*}
if and only if $\mathcal{C}_{i_{r}}\subseteq\mathcal{C}_{j_{r}}^{\bot_{H}}$ for all $r\in\{1,\ldots,s\}$ and there exists a unique $j\in\mathcal{F}$ such that
$\mathcal{C}_{j}$ is AHSO and $\mathcal{C}_{i}$ is HSO for all $i\in\mathcal{F}\backslash\{j\}$.

\vspace{-4pt}

\item [5)] $\mathcal{C}_{A,k}$ is Hermitian LCD if and only if
$\mathcal{C}_{i_{r}}\cap\mathcal{C}_{j_{r}}^{\bot_{H}}=\{\mathbf{0}\}$ and $\mathcal{C}_{j_{r}}\cap\mathcal{C}_{i_{r}}^{\bot_{H}}=\{\mathbf{0}\}$
for all $r\in\{1,\ldots,s\}$ and $\mathcal{C}_{i}$ is Hermitian LCD for all $i\in\mathcal{F}$.
\end{itemize}
\end{theorem}

\begin{proof}
1) Statement 1) follows directly.

2) We have that $\mathcal{C}_{A,k}$ is AHDC if and only if
\vspace{-6pt}
\begin{align*}
\sum_{r=1}^{s}\big(\mathrm{dim}(\mathcal{C}_{j_{r}}^{\bot_{H}})-\mathrm{dim}(\mathcal{C}_{i_{r}}\cap\mathcal{C}_{j_{r}}^{\bot_{H}})\big)
+\sum_{r=1}^{s}\big(\mathrm{dim}(\mathcal{C}_{i_{r}}^{\bot_{H}})-\mathrm{dim}(\mathcal{C}_{j_{r}}\cap\mathcal{C}_{i_{r}}^{\bot_{H}})\big) \\
+\sum_{i\in\mathcal{F}}\big(\mathrm{dim}(\mathcal{C}_{i}^{\bot_{H}})-\mathrm{dim}(\mathcal{C}_{i}\cap\mathcal{C}_{i}^{\bot_{H}})\big)=1,
\end{align*}
which, together with the result
$\mathrm{dim}(\mathcal{C}_{j_{r}}\cap\mathcal{C}_{i_{r}}^{\bot_{H}})=\mathrm{dim}(\mathcal{C}_{i_{r}}\cap\mathcal{C}_{j_{r}}^{\bot_{H}})
+\mathrm{dim}(\mathcal{C}_{j_{r}})-\mathrm{dim}(\mathcal{C}_{i_{r}})$
derived from Lemma 3.1, implies that it is equivalent to
\vspace{-6pt}
\begin{align}
2\sum_{r=1}^{s}\big(\mathrm{dim}(\mathcal{C}_{j_{r}}^{\bot_{H}})-\mathrm{dim}(\mathcal{C}_{i_{r}}\cap\mathcal{C}_{j_{r}}^{\bot_{H}})\big)
+\sum_{i\in\mathcal{F}}\big(\mathrm{dim}(\mathcal{C}_{i}^{\bot_{H}})-\mathrm{dim}(\mathcal{C}_{i}\cap\mathcal{C}_{i}^{\bot_{H}})\big)=1.
\end{align}
This is also equivalent to
\vspace{-6pt}
\begin{align*}
\sum_{r=1}^{s}\big(\mathrm{dim}(\mathcal{C}_{j_{r}}^{\bot_{H}})-\mathrm{dim}(\mathcal{C}_{i_{r}}\cap\mathcal{C}_{j_{r}}^{\bot_{H}})\big)=0,\
\sum_{i\in\mathcal{F}}\big(\mathrm{dim}(\mathcal{C}_{i}^{\bot_{H}})-\mathrm{dim}(\mathcal{C}_{i}\cap\mathcal{C}_{i}^{\bot_{H}})\big)=1.
\end{align*}
Namely, $\mathcal{C}_{j_{r}}^{\bot_{H}}\subseteq\mathcal{C}_{i_{r}}$ for all $r\in\{1,\ldots,s\}$, and
there exists a unique $j\in\mathcal{F}$ such that
$\mathrm{dim}(\mathcal{C}_{j}^{\bot_{H}})-\mathrm{dim}(\mathcal{C}_{j}\cap\mathcal{C}_{j}^{\bot_{H}})=1$, i.e., $\mathcal{C}_{j}$ is AHDC,
and $\mathcal{C}_{i}^{\bot_{H}}\subseteq\mathcal{C}_{i}$, i.e., $\mathcal{C}_{i}$ is HDC for all $i\in\mathcal{F}\backslash\{j\}$.
This verifies statement 2).

3) Statement 3) follows directly.

4) We have that $\mathcal{C}_{A,k}$ is AHSO if and only if
\vspace{-6pt}
\begin{align*}
\sum_{r=1}^{s}\big(\mathrm{dim}(\mathcal{C}_{i_{r}})-\mathrm{dim}(\mathcal{C}_{i_{r}}\cap\mathcal{C}_{j_{r}}^{\bot_{H}})\big)
+\sum_{r=1}^{s}\big(\mathrm{dim}(\mathcal{C}_{j_{r}})-\mathrm{dim}(\mathcal{C}_{j_{r}}\cap\mathcal{C}_{i_{r}}^{\bot_{H}})\big) \\
+\sum_{i\in\mathcal{F}}\big(\mathrm{dim}(\mathcal{C}_{i})-\mathrm{dim}(\mathcal{C}_{i}\cap\mathcal{C}_{i}^{\bot_{H}})\big)=1,
\end{align*}
which, together with the result
$\mathrm{dim}(\mathcal{C}_{j_{r}}\cap\mathcal{C}_{i_{r}}^{\bot_{H}})=\mathrm{dim}(\mathcal{C}_{i_{r}}\cap\mathcal{C}_{j_{r}}^{\bot_{H}})
+\mathrm{dim}(\mathcal{C}_{j_{r}})-\mathrm{dim}(\mathcal{C}_{i_{r}})$
derived from Lemma 3.1, implies that it is equivalent to
\vspace{-6pt}
\begin{align}
2\sum_{r=1}^{s}\big(\mathrm{dim}(\mathcal{C}_{i_{r}})-\mathrm{dim}(\mathcal{C}_{i_{r}}\cap\mathcal{C}_{j_{r}}^{\bot_{H}})\big)
+\sum_{i\in\mathcal{F}}\big(\mathrm{dim}(\mathcal{C}_{i})-\mathrm{dim}(\mathcal{C}_{i}\cap\mathcal{C}_{i}^{\bot_{H}})\big)=1.
\end{align}
This is also equivalent to
\vspace{-6pt}
\begin{align*}
\sum_{r=1}^{s}\big(\mathrm{dim}(\mathcal{C}_{i_{r}})-\mathrm{dim}(\mathcal{C}_{i_{r}}\cap\mathcal{C}_{j_{r}}^{\bot_{H}})\big)=0,\
\sum_{i\in\mathcal{F}}\big(\mathrm{dim}(\mathcal{C}_{i})-\mathrm{dim}(\mathcal{C}_{i}\cap\mathcal{C}_{i}^{\bot_{H}})\big)=1.
\end{align*}
Namely, $\mathcal{C}_{i_{r}}\subseteq\mathcal{C}_{j_{r}}^{\bot_{H}}$ for all $r\in\{1,\ldots,s\}$, and
there exists a unique $j\in\mathcal{F}$ such that
$\mathrm{dim}(\mathcal{C}_{j})-\mathrm{dim}(\mathcal{C}_{j}\cap\mathcal{C}_{j}^{\bot_{H}})=1$, i.e., $\mathcal{C}_{j}$ is AHSO,
and $\mathcal{C}_{i}\subseteq\mathcal{C}_{i}^{\bot_{H}}$, i.e., $\mathcal{C}_{i}$ is HSO for all $i\in\mathcal{F}\backslash\{j\}$.
This verifies statement 4).

5) Statement 5) follows directly. $\hfill\square$
\end{proof}

\begin{table}[!htbp]
\renewcommand\arraystretch{1.4}
\centering	
\small
\setlength{\abovecaptionskip}{0.cm}
\setlength{\belowcaptionskip}{0.2cm}
\caption{All the possible manners to derive a HDC MP code $\mathcal{C}_{A,k}$ for $k=2,3,4$.}
\vspace{3pt}
\begin{tabular}{|c|c|c|c|}
\hline
$k$&Methods&Constituent codes $\mathcal{C}_{1},\ldots,\mathcal{C}_{k}$&Derived MP code $\mathcal{C}_{A,k}$\\
\hline
\hline
\multirow{3}*{2}&\makecell[c]{Class 1\\ \footnotesize{$\tau=(1)$}}
&$\mathcal{C}_{i}^{\bot_{H}}\subseteq \mathcal{C}_{i}$, $i=1,2$&\multirow{3}*{$\mathcal{C}_{A,2}^{\bot_{H}}\subseteq \mathcal{C}_{A,2}$}\\
\cline{2-3}
~&\makecell[c]{Class 2\\ \footnotesize{$\tau=(1\ 2)$}}&$\mathcal{C}_{1}^{\bot_{H}}\subseteq \mathcal{C}_{2}$&~\\
\hline
\hline
\multirow{6}*{3}&\makecell[c]{Class 1\\ \footnotesize{$\tau=(1)$}}
&$\mathcal{C}_{i}^{\bot_{H}}\subseteq \mathcal{C}_{i}$, $i=1,2,3$&\multirow{6}*{$\mathcal{C}_{A,3}^{\bot_{H}}\subseteq \mathcal{C}_{A,3}$}\\
\cline{2-3}
~&\makecell[c]{Class 2\\ \footnotesize{$\tau=(1\ 2)$}}&\makecell[c]{$\mathcal{C}_{1}^{\bot_{H}}\subseteq \mathcal{C}_{2}$,\\ $\mathcal{C}_{3}^{\bot_{H}}\subseteq \mathcal{C}_{3}$}&~\\
\cline{2-3}
~&\makecell[c]{Class 3\\ \footnotesize{$\tau=(1\ 3)$}}&\makecell[c]{$\mathcal{C}_{1}^{\bot_{H}}\subseteq \mathcal{C}_{3}$,\\ $\mathcal{C}_{2}^{\bot_{H}}\subseteq \mathcal{C}_{2}$}&~\\
\cline{2-3}
~&\makecell[c]{Class 4\\ \footnotesize{$\tau=(2\ 3)$}}&\makecell[c]{$\mathcal{C}_{2}^{\bot_{H}}\subseteq \mathcal{C}_{3}$,\\ $\mathcal{C}_{1}^{\bot_{H}}\subseteq \mathcal{C}_{1}$}&~\\
\hline
\hline
\multirow{15}*{4}&\makecell[c]{Class 1\\ \footnotesize{$\tau=(1)$}}
&$\mathcal{C}_{i}^{\bot_{H}}\subseteq \mathcal{C}_{i}$, $i=1,2,3,4$&\multirow{15}*{$\mathcal{C}_{A,4}^{\bot_{H}}\subseteq \mathcal{C}_{A,4}$}\\
\cline{2-3}
~&\makecell[c]{Class 2\\ \footnotesize{$\tau=(1\ 2)$}}&\makecell[c]{$\mathcal{C}_{1}^{\bot_{H}}\subseteq \mathcal{C}_{2}$,\\$\mathcal{C}_{i}^{\bot_{H}}\subseteq \mathcal{C}_{i}$, $i=3,4$}&~\\
\cline{2-3}
~&\makecell[c]{Class 3\\ \footnotesize{$\tau=(1\ 3)$}}&\makecell[c]{$\mathcal{C}_{1}^{\bot_{H}}\subseteq \mathcal{C}_{3}$,\\ $\mathcal{C}_{i}^{\bot_{H}}\subseteq \mathcal{C}_{i}$, $i=2,4$}&~\\
\cline{2-3}
~&\makecell[c]{Class 4\\ \footnotesize{$\tau=(1\ 4)$}}&\makecell[c]{$\mathcal{C}_{1}^{\bot_{H}}\subseteq \mathcal{C}_{4}$,\\ $\mathcal{C}_{i}^{\bot_{H}}\subseteq \mathcal{C}_{i}$, $i=2,3$}&~\\
\cline{2-3}
~&\makecell[c]{Class 5\\ \footnotesize{$\tau=(2\ 3)$}}&\makecell[c]{$\mathcal{C}_{2}^{\bot_{H}}\subseteq \mathcal{C}_{3}$,\\ $\mathcal{C}_{i}^{\bot_{H}}\subseteq \mathcal{C}_{i}$, $i=1,4$}&~\\
\cline{2-3}
~&\makecell[c]{Class 6\\ \footnotesize{$\tau=(2\ 4)$}}&\makecell[c]{$\mathcal{C}_{2}^{\bot_{H}}\subseteq \mathcal{C}_{4}$,\\ $\mathcal{C}_{i}^{\bot_{H}}\subseteq \mathcal{C}_{i}$, $i=1,3$}&~\\
\cline{2-3}
~&\makecell[c]{Class 7\\ \footnotesize{$\tau=(3\ 4)$}}&\makecell[c]{$\mathcal{C}_{3}^{\bot_{H}}\subseteq \mathcal{C}_{4}$,\\ $\mathcal{C}_{i}^{\bot_{H}}\subseteq \mathcal{C}_{i}$, $i=1,2$}&~\\
\cline{2-3}
~&\makecell[c]{Class 8\\ \footnotesize{$\tau=(1\ 2)(3\ 4)$}}&\makecell[c]{$\mathcal{C}_{1}^{\bot_{H}}\subseteq \mathcal{C}_{2}$,\\ $\mathcal{C}_{3}^{\bot_{H}}\subseteq \mathcal{C}_{4}$}&~\\
\cline{2-3}
~&\makecell[c]{Class 9\\ \footnotesize{$\tau=(1\ 3)(2\ 4)$}}&\makecell[c]{$\mathcal{C}_{1}^{\bot_{H}}\subseteq \mathcal{C}_{3}$,\\ $\mathcal{C}_{2}^{\bot_{H}}\subseteq \mathcal{C}_{4}$}&~\\
\cline{2-3}
~&\makecell[c]{Class 10\\ \footnotesize{$\tau=(1\ 4)(2\ 3)$}}&\makecell[c]{$\mathcal{C}_{1}^{\bot_{H}}\subseteq \mathcal{C}_{4}$,\\ $\mathcal{C}_{2}^{\bot_{H}}\subseteq \mathcal{C}_{3}$}&~\\
\hline
\end{tabular}
\end{table}

\subsection{Another characterization of AHDC and AHSO MP codes}

\begin{theorem}
With the notations in Theorem 4.4. Then the following statements hold.

i) $\mathcal{C}_{A,k}$ is AHDC if and only if
\vspace{-7pt}
\begin{align*}
2\sum_{r=1}^{s}\big(\mathrm{dim}(\mathcal{C}_{j_{r}})+\mathrm{dim}(\mathcal{C}_{i_{r}}\cap\mathcal{C}_{j_{r}}^{\bot_{H}})\big)
+\sum_{i\in\mathcal{F}}\big(\mathrm{dim}(\mathcal{C}_{i})+\mathrm{dim}(\mathcal{C}_{i}\cap\mathcal{C}_{i}^{\bot_{H}})\big)
=kn-1\Big.
\end{align*}

ii) $\mathcal{C}_{A,k}$ is AHSO if and only if
\vspace{-7pt}
\begin{align*}
2\sum_{r=1}^{s}\big(\mathrm{dim}(\mathcal{C}_{i_{r}}^{\bot_{H}})+\mathrm{dim}(\mathcal{C}_{i_{r}}\cap\mathcal{C}_{j_{r}}^{\bot_{H}})\big)
+\sum_{i\in\mathcal{F}}\big(\mathrm{dim}(\mathcal{C}_{i}^{\bot_{H}})+\mathrm{dim}(\mathcal{C}_{i}\cap\mathcal{C}_{i}^{\bot_{H}})\big)
=kn-1\Big.
\end{align*}
\end{theorem}

\begin{proof}
By $\mathrm{dim}(\mathcal{C}_{j_{r}}^{\bot_{H}})=n-\mathrm{dim}(\mathcal{C}_{j_{r}})$ and
$\mathrm{dim}(\mathcal{C}_{i}^{\bot_{H}})=n-\mathrm{dim}(\mathcal{C}_{i})$, statement i) then follows.
Similarly, by $\mathrm{dim}(\mathcal{C}_{i_{r}})=n-\mathrm{dim}(\mathcal{C}_{i_{r}}^{\bot_{H}})$ and
$\mathrm{dim}(\mathcal{C}_{i})=n-\mathrm{dim}(\mathcal{C}_{i}^{\bot_{H}})$, statement ii) then follows. $\hfill\square$
\end{proof}

\begin{corollary}
With the notations in Theorem 4.5. Then the following statements hold.

1) There are no AHDC MP codes $\mathcal{C}_{A,k}$ if one of the following cases holds:
\vspace{-3pt}
\begin{itemize}
\item $k$ and $n$ odd,
$\sum_{i\in\mathcal{F}}\big(\mathrm{dim}(\mathcal{C}_{i})+\mathrm{dim}(\mathcal{C}_{i}\cap\mathcal{C}_{i}^{\bot_{H}})\big)$ odd;

\item $k$ or $n$ even, $\sum_{i\in\mathcal{F}}\big(\mathrm{dim}(\mathcal{C}_{i})+\mathrm{dim}(\mathcal{C}_{i}\cap\mathcal{C}_{i}^{\bot_{H}})\big)$ even.
\end{itemize}

2) There are no AHSO MP codes $\mathcal{C}_{A,k}$ if one of the following cases holds:
\vspace{-3pt}
\begin{itemize}
\item $k$ and $n$ odd, $\sum_{i\in\mathcal{F}}\big(\mathrm{dim}(\mathcal{C}_{i}^{\bot_{H}})+\mathrm{dim}(\mathcal{C}_{i}\cap\mathcal{C}_{i}^{\bot_{H}})\big)$ odd;

\item $k$ or $n$ even, $\sum_{i\in\mathcal{F}}\big(\mathrm{dim}(\mathcal{C}_{i}^{\bot_{H}})+\mathrm{dim}(\mathcal{C}_{i}\cap\mathcal{C}_{i}^{\bot_{H}})\big)$ even.
\end{itemize}
\end{corollary}

\vspace{6pt}

By Corollary 4.6, we can directly verify the following result.

\begin{corollary}
With the notations in Theorem 4.4. Then the following statements hold.

1) There are no AHDC MP codes $\mathcal{C}_{A,k}$ if one of the following cases holds.

1.1) $\mathcal{F}=\varnothing$.

1.2) $\mathcal{F}\neq\varnothing$, $\mathcal{C}_{i}\subseteq \mathcal{C}_{i}^{\bot_{H}}$ for $i\in \mathcal{E}\subseteq\mathcal{F}$, and one of the following cases holds:
\vspace{-3pt}
\begin{itemize}
\item $k$ and $n$ odd,
$\sum_{i\in\mathcal{F}\backslash\mathcal{E}}\big(\mathrm{dim}(\mathcal{C}_{i})+\mathrm{dim}(\mathcal{C}_{i}\cap\mathcal{C}_{i}^{\bot_{H}})\big)$ odd;

\vspace{-3pt}
\item $k$ or $n$ even, $\sum_{i\in\mathcal{F}\backslash\mathcal{E}}\big(\mathrm{dim}(\mathcal{C}_{i})+\mathrm{dim}(\mathcal{C}_{i}\cap\mathcal{C}_{i}^{\bot_{H}})\big)$ even.
\end{itemize}
\vspace{-4pt}
In particular, for case 1.2), if $\mathcal{E}=\mathcal{F}\neq\varnothing$, and $k$ or $n$ even, then there are no AHDC MP codes $\mathcal{C}_{A,k}$.

2) There are no AHSO MP codes $\mathcal{C}_{A,k}$ if one of the following cases holds.

2.1) $\mathcal{F}=\varnothing$.

2.2) $\mathcal{F}\neq\varnothing$, $\mathcal{C}_{i}^{\bot_{H}}\subseteq \mathcal{C}_{i}$ for $i\in \mathcal{E}\subseteq\mathcal{F}$, and one of the following cases holds:
\vspace{-3pt}
\begin{itemize}
\item $k$ and $n$ odd, $\sum_{i\in\mathcal{F}\backslash\mathcal{E}}\big(\mathrm{dim}(\mathcal{C}_{i}^{\bot_{H}})
+\mathrm{dim}(\mathcal{C}_{i}\cap\mathcal{C}_{i}^{\bot_{H}})\big)$ odd;

\vspace{-3pt}
\item $k$ or $n$ even, $\sum_{i\in\mathcal{F}\backslash\mathcal{E}}\big(\mathrm{dim}(\mathcal{C}_{i}^{\bot_{H}})
+\mathrm{dim}(\mathcal{C}_{i}\cap\mathcal{C}_{i}^{\bot_{H}})\big)$ even.
\end{itemize}
\vspace{-4pt}
In particular, for case 2.2), if $\mathcal{E}=\mathcal{F}\neq\varnothing$, and $k$ or $n$ even, then there are no AHDC MP codes $\mathcal{C}_{A,k}$.
\end{corollary}

\section{HDC MP codes}

For the trivial case, i.e., $AA^{\dag}$ is invertible diagonal, we obtain the following method for constructing HDC and AHDC MP codes.

\begin{theorem}
(Trivial case)
Let $AA^{\dag}$ be invertible diagonal for $A\in\mathcal{M}_{k}(\mathbb{F}_{q^{2}})$.
Let $\mathcal{C}_{A,k}=[\mathcal{C}_{1},\mathcal{C}_{2},\ldots,\mathcal{C}_{k}]\cdot A$,
where $\mathcal{C}_{i}$ is an $[n,t_{i},d_{i}]_{q^{2}}$ code for $i=1,\ldots,k$.
Then the following statements hold.
\begin{itemize}
\item [1)] If $\mathcal{C}_{i}$ is HDC for all $i\in\{1,\ldots,k\}$, then $\mathcal{C}_{A,k}$ is a
$\big[kn,\sum_{i=1}^{k}t_{i},\geq\min\limits_{\scriptscriptstyle 1\leq i\leq k}\{D_{i}(A)d_{i}\}\big]_{q^{2}}$ HDC code.
Furthermore, if $A$ is also NSC, then $\mathcal{C}_{A,k}$ is a
$\big[kn,\sum_{i=1}^{k}t_{i},\geq\min\limits_{\scriptscriptstyle 1\leq i\leq k}\{(k-i+1)d_{i}\}\big]_{q^{2}}$ HDC code.

\vspace{-4pt}

\item [2)] If there is a unique $j\in\{1,\ldots,k\}$ such that $\mathcal{C}_{j}$ is AHDC and $\mathcal{C}_{i}$ is HDC for all $i\in\{1,\ldots,k\}\backslash\{j\}$,
then $\mathcal{C}_{A,k}$ is a $[kn,\sum_{i=1}^{k}t_{i},\geq\min\limits_{\scriptscriptstyle 1\leq i\leq k}\{D_{i}(A)d_{i}\}\big]_{q^{2}}$ AHDC code.
Furthermore, if $A$ is also NSC, then $\mathcal{C}_{A,k}$ is a
$\big[kn,\sum_{i=1}^{k}t_{i},\geq\min\limits_{\scriptscriptstyle 1\leq i\leq k}\{(k-i+1)d_{i}\}\big]_{q^{2}}$ AHDC code.
\end{itemize}
\end{theorem}

\begin{proof}
According to previous lemmas, the result follows. $\hfill\square$
\end{proof}

\vspace{5pt}

For the nontrivial case, i.e., $AA^{\dag}$ is $\tau$-monomial with $\tau\neq\mathbbm{1}_{k}$, we give the corresponding method
for constructing HDC and AHDC MP codes in the following theorem.

\begin{theorem}
(Nontrivial case)
Let $AA^{\dag}$ be $\tau$-monomial for $A\in\mathcal{M}_{k}(\mathbb{F}_{q^{2}})$, where $\tau=(i_{1}\ j_{1})(i_{2}\ j_{2})\cdots(i_{s}\ j_{s})$
denotes the product of $s$ disjoint 2-cycles for $s\in\{1,\ldots,\lfloor\frac{k}{2}\rfloor\}$.
Let $\mathcal{F}=\{1,\ldots,k\}-\{i_{1},\ldots,i_{s},j_{1},\ldots,j_{s}\}$.
Let $\mathcal{C}_{A,k}=[\mathcal{C}_{1},\mathcal{C}_{2},\ldots,\mathcal{C}_{k}]\cdot A$, where $\mathcal{C}_{i}$ is an $[n,t_{i},d_{i}]_{q^{2}}$ code for $i=1,\ldots,k$.
Then the following statements hold.
\begin{itemize}
\item [1)] If $\mathcal{C}_{i_{r}}^{\bot_{H}}\subseteq\mathcal{C}_{j_{r}}$ for all $r\in\{1,\ldots,s\}$
and $\mathcal{C}_{i}$ is HDC for all $i\in\mathcal{F}$, then $\mathcal{C}_{A,k}$ is a
$\big[kn,\sum_{i=1}^{k}t_{i},\geq\min\limits_{\scriptscriptstyle 1\leq i\leq k}\{D_{i}(A)d_{i}\}\big]_{q^{2}}$ HDC code.
Furthermore, if $A$ is also NSC, then $\mathcal{C}_{A,k}$ is a
$\big[kn,\sum_{i=1}^{k}t_{i},\geq\min\limits_{\scriptscriptstyle 1\leq i\leq k}\{(k-i+1)d_{i}\}\big]_{q^{2}}$ HDC code.

\vspace{-4pt}

\item [2)] If $\mathcal{F}\neq\varnothing$, $\mathcal{C}_{i_{r}}^{\bot_{H}}\subseteq\mathcal{C}_{j_{r}}$ for all $r\in\{1,\ldots,s\}$,
and there is a unique $j\in\mathcal{F}$ such that $\mathcal{C}_{j}$ is AHDC and $\mathcal{C}_{i}$ is HDC for all $i\in\mathcal{F}\backslash\{j\}$,
then $\mathcal{C}_{A,k}$ is a $[kn,\sum_{i=1}^{k}t_{i},\geq\min\limits_{\scriptscriptstyle 1\leq i\leq k}\{D_{i}(A)d_{i}\}\big]_{q}$ AHDC code.
Furthermore, if $A$ is also NSC, then $\mathcal{C}_{A,k}$ is a
$\big[kn,\sum_{i=1}^{k}t_{i},\geq\min\limits_{\scriptscriptstyle 1\leq i\leq k}\{(k-i+1)d_{i}\}\big]_{q}$ AHDC code.
\end{itemize}
\end{theorem}

\begin{proof}
According to previous lemmas, the result follows. $\hfill\square$
\end{proof}

\section{Conclusion}

In this paper, we provided an explicit formula for calculating the dimension of the Hermitian hull of a MP code.
We presented the necessary and sufficient conditions for a MP code to be HDC, AHDC, HSO, AHSO and HLCD, respectively.
We theoretically determined the number of all possible ways involving the relationships among the constituent codes to yield a MP code
that is HDC, AHDC, HSO and AHSO, respectively.
We gave alternative necessary and sufficient conditions for a MP code to be AHDC and AHSO, respectively,
and provided several cases where a MP code is not AHDC or AHSO.
We gave the construction methods of HDC and AHDC MP codes.

We emphasize that in the near future, we will delve deeper into the matrices considered in this paper and provide further results regarding these matrices.

\footnotesize{
\bibliographystyle{plain}
\phantomsection
\addcontentsline{toc}{section}{References}
\bibliography{2024.5.11}
}

\end{document}